\documentclass[
 onecolumn,
 11pt, 
 draftcls,
journal
]{IEEEtran}

\usepackage{array}
\usepackage{epsfig}
\usepackage{cite}
\usepackage{amsfonts}
\usepackage{amsmath}
\usepackage{amssymb}
\usepackage{amsxtra}
\usepackage{varioref}
\usepackage{multicol}
\usepackage{float}
\usepackage{rotate}
\usepackage{color}
\usepackage{enumerate}
\usepackage[active]{srcltx} 
\usepackage{mathrsfs}
\usepackage{color}
\usepackage{rotating}
\usepackage{Milancommands}
\usepackage{mathtools}
\def\independenT#1#2{\mathrel{\rlap{$#1#2$}\mkern2mu{#1#2}}}

\usepackage{times}
\usepackage[english]{babel}
\usepackage{inputenc}
\usepackage{fontenc}

\newcommand{\BibPath}{/home/milan/Documents/University/Research/BibTeX}
\newtheorem{thm}{\textbf{Theorem}}
\newtheorem{defn}{\textbf{Definition}}
\newtheorem{coro}{\textbf{Corollary}}
\newtheorem{lem}{Lemma}
\newtheorem{rem}{Remark}

\title{
Fundamental Inequalities and Identities Involving Mutual and Directed Informations in Closed-Loop Systems
} 

 \author{ 
Milan S. Derpich, Eduardo I. Silva  and Jan {\O}stergaard
 \thanks{
 M.S. Derpich and E.I Silva are with the Department of Electronic Engineering, Universidad T\'ecnica Federico Santa Mar\'ia, Casilla 110-V, Valpara\'iso, Chile
 (email: milan.derpich@usm.cl, eduardo.silva@usm.cl).
Their work was supported in part by CONICYT through grants FONDECYT 
Nr.~1120468,
Nr.~1130459,
and Anillo ACT-53.   
}
 \thanks{%
J. {\O}stergaard is with the Department of Electronic Systems, Aalborg University, Niels Jernes Vej
12, DK-9220, Aalborg, Denmark (email: janoe@ieee.org).
}
}

\begin{document}
\maketitle
\begin{abstract} 
We present several novel identities and
inequalities relating the mutual information and the directed
information in systems with feedback. 
The internal blocks within such systems are restricted only to be causal mappings, but
are allowed to be non-linear, stochastic and time varying. 
Moreover, the involved signals can be arbitrarily distributed. 
We bound the directed information
between signals inside the feedback loop by the mutual information
between signals inside and outside the feedback loop. 
This fundamental result has an interesting interpretation as a law of conservation of
information flow. 
Building upon it, 
we derive several novel identities and inequalities, which allow us to
prove some existing information inequalities under less restrictive assumptions.
Finally, we establish new relationships between nested directed informations inside a feedback loop. 
This yields a new and general data-processing inequality for systems with feedback.

\end{abstract}

\section{Introduction}\label{sec:intro}
The notion of directed information introduced by Massey in~\cite{massey90} assesses the amount of information that causally ``flows'' from a given random and ordered sequence to another. 
For this reason, it has increasingly found use in diverse applications, from
characterizing the capacity of channels with feedback~\cite{massey90,kramer98,tatmit09,li-eli11},
the rate distortion function under causality constraints~\cite{derost12},
establishing some of the fundamental limitations in networked control~\cite{tatiko00,mardah05,mardah08,silder11,silder10,silder11b}, 
determining causal relationships in neural networks~\cite{quicol11},
to 
portfolio theory and hypothesis testing~\cite{perkim11}, to name a few.

The directed information from a random%
\footnote{
Hereafter we use non-italic letters (such as $\rvax$) for random variables, denoting a particular realization by the corresponding italic character, $x$.
}
 sequence $\rvax^{k}$ to a random sequence  $\rvay^{k}$ is defined as 
\begin{align}\label{eq:directed_inf_def}
 I(\rvax^{k}\to\rvay^{k}) \eq \Sumfromto{i=1}{k}I(\rvay(i);\rvax^{i}|\rvay^{i-1}),
\end{align}
where the notation $\rvax^{i}$ represents the sequence $\rvax(1),\rvax(2),\ldots, \rvax(i)$.
The causality inherent in this definition becomes evident when comparing it with the mutual information between $\rvax^{k}$ and $\rvay^{k}$, given by $I(\rvax^{k},\rvay^{k})=\sumfromto{i=1}{k}I(\rvay(i);\rvax^{k}|\rvay^{i-1})$.
In the latter sum, what matters is the amount of information about the \textit{entire} sequence $\rvax^{k}$
present in $\rvay(i)$, given the past values $\rvay^{i-1}$.
By contrast, in the conditional mutual informations in the sum of~\eqref{eq:directed_inf_def}, only the past and current values of $\rvax^{k}$ are considered, that is, $\rvax^{i}$.
Thus,  $I(\rvax^{k}\to\rvay^{k})$ represents the amount of information causally conveyed from $\rvax^{k}$ to $\rvay^{k}$.

There exist several results characterizing the relationship between $I(\rvax^{k}\to\rvay^{k})$ and $I(\rvax^{k};\rvay^{k})$.
First, it is well known that $I(\rvax^{k}\to\rvay^{k})\leq I(\rvax^{k};\rvay^{k})$, with equality if and only if $\rvay^{k}$ is causally related to $\rvax^{k}$%
~\cite{massey90}.
A conservation law of mutual and directed information has been found in~\cite{masmas05}, which asserts that 
$I(\rvax^{k}\to\rvay^{k}) + I(0*\rvay^{k-1}\to \rvax^{k}) = I(\rvax^{k};\rvay^{k})$, where $0*\rvay^{k-1}$ denotes the concatenation $0,\rvay(1),\ldots,\rvay^{k-1}$.

Given its prominence in settings involving feedback, it is perhaps in these scenarios where the directed information becomes most important.
For instance, the directed information has been instrumental in characterizing the capacity of channels with feedback (see, e.g.,~\cite{tatmit09,kim-yh08,li-eli11} and the references therein), as well as the rate-distortion function in setups involving feedback~\cite{zamkoc08,silder11,silder11b,silder10,derost12}.

In this paper, our focus is on the relationships (inequalities and identities) involving directed and mutual informations within feedback systems, as well as between directed informations involving different signals within the corresponding feedback loop.
In order to discuss some of the existing results related to this problem,
it is convenient to consider the general feedback system shown in Fig.~\ref{fig:diagramas}-(a).
In this diagram, the blocks $\Ssp_{1},\ldots, \Ssp_{4}$ represent possibly non-linear and time-varying causal systems such that the total delay of the loop is at least one sample. 
In the same figure, $\rvar,\rvap,\rvas,\rvaq$ are exogenous random signals (scalars, vectors or sequences), which could represent, for example, any combination of disturbances, noises, random initial states or side informations.
We note that any of these exogenous signals, in combination with its corresponding deterministic mapping $\Ssp_{i}$, can also yield any desired stochastic causal mapping.

For the simple case in which all the systems $\set{\Ssp_{i}}_{i=1}^{4}$ are linear time invariant (LTI) and stable, and assuming $\rvap,\rvax,\rvaq=0$ (deterministically), it was shown in~\cite{zhasun06} that $I(\rvar^{k}\to \rvae^{k})$ does not depend on whether there is feedback from $\rvae$ to $\rvau$ or not.
Inequalities between mutual and directed informations in a less restricted setup, shown in Fig.~\ref{fig:diagramas}-(b), have been found in~\cite{mardah05,mardah08}.  
In that setting (a networked-control system), $G$ is a strictly causal LTI dynamic system having 
(vector) state sequence $\set{\rvex(i)}_{i=0}^{\infty}$, with 
$\rvex_{0}\eq\rvex(0)$ being the random initial state in its state-space representation.
The external signal $\rvar$ (which could correspond to a disturbance) is statistically independent of $\rvas$, the latter corresponding to, for example, side information or channel noise.
Both are also statistically independent of $\rvex_{0}$.
\begin{figure}[htbp]
\centering
\input{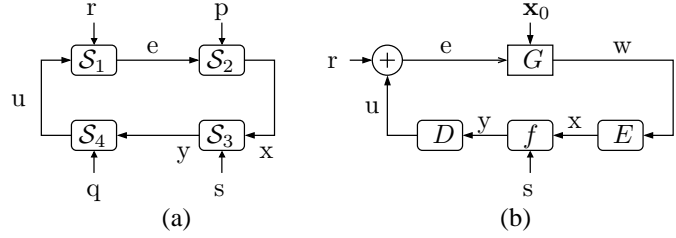}
\caption{(a): The general system considered in this work. 
(b): A special case, corresponding to the closed-loop system studied in~\cite{mardah05}.} 
\label{fig:diagramas}
\end{figure}
The blocks labeled $E$, $D$ and $f$ correspond to an encoder, a decoder and a channel, respectively, all of which are causal.
The channel $f$ maps $\rvas^{k}$ and $\rvax^{k}$ to $\rvay(k)$ in a possibly time-varying manner, i.e.,
$
 \rvay(k)=f(k,\rvax^{k},\rvas^{k}).
$
Similarly, the concatenation of the encoder, the channel and the decoder, maps $\rvas^{k}$ and $\rvaw^{k}$ to $\rvau(k)$ as a possibly time-dependent function
$
 \rvau(k)=\psi(k,\rvaw^{k},\rvas^{k}).
$
Under these assumptions, the following fundamental result was shown in~\cite[Lemma~5.1]{mardah08}:
\begin{align}\label{eq:Martins_first}
 I(\rvex_{0},\rvar^{k}\,;\,\rvau^{k})
-
I(\rvar^{k}; \rvau^{k})
\geq 
I(\rvex_{0};\rvae^{k}).
\end{align}
%
By further assuming in~\cite{mardah08} that the decoder $D$ in Fig.~\ref{fig:diagramas}-(b) is deterministic,
the following Markov chain naturally holds,
\begin{align}\label{eq:MC_martins}
 (\rvex_{0},\rvar^{k}) 
\longleftrightarrow
\rvay^{k}
\longleftrightarrow
\rvau^{k},
\end{align}
leading directly to 
%
\begin{align}\label{eq:la_de_martins08}
  I(\rvex_{0},\rvar^{k}\,;\,\rvay^{k})
-
I(\rvar^{k}; \rvau^{k})
\geq 
I(\rvex_{0};\rvae^{k}),
\end{align}
which is found in the proof of~\cite[Corollary~5.3]{mardah08}.
The deterministic nature of the decoder $D$ played a crucial role in the proof of this result,  since otherwise the Markov chain~\eqref{eq:MC_martins} does not hold, in general, due to the feedback from $\rvau$ to $\rvay$. 

Notice that both~\eqref{eq:Martins_first} and~\eqref{eq:la_de_martins08} provide lower bounds to the difference between two mutual informations, each of them relating a signal \textit{external} to the loop (such as $\rvex_{0},\rvar^{k}$) to a signal \textit{internal} to the loop (such as $\rvau^{k}$ or $\rvay^{k}$).
Instead, the inequality 
\begin{align}\label{eq:la_de_Massey}
I(\rvax^{k}\to\rvay^{k})\geq I(\rvar^{k};\rvay^{k}),
\end{align}
which holds for the system in Fig.~\ref{fig:diagramas}-(a) and appears 
in~\cite[Theorem~3]{massey90} (and rediscovered later in~\cite[Lemma~4.8.1]{tatiko00}),
involves the directed information between two internal signals and the mutual information between the second of these and an external sequence.
A related bound,  similar to~\eqref{eq:la_de_martins08} 
but involving information rates and with the leftmost mutual information 
replaced by the directed information from $\rvax^{k}$ to $\rvay^{k}$ (which are two signals internal to the loop), has been obtained in~\cite[Lemma~4.1]{mardah05}:
\begin{align}\label{eq:mardah_dir_minus_mutual}
 \bar{I}(\rvax\to \rvay) - \bar{I}(\rvar;\rvau) \geq \lim_{k\to\infty}\frac{I(\rvex(0);\rvae^{k})}{k},
\end{align}
with $\bar{I}(\rvax\to \rvay)\eq \lim_{k\to\infty}\frac{1}{k}I(\rvax^{k}\to \rvay^{k})$
and
$\bar{I}(\rvar;\rvau)\eq \lim_{k\to\infty}\frac{1}{k}I(\rvar^{k};\rvau^{k})$,
provided $\sup_{i\geq 0}\Expe{\rvex(i)^{T}\rvex(i)}<\infty$.
This result relies on
three assumptions: 
a) that the channel $f$ is memory-less and satisfies a ``conditional invertibility'' property, b) a finite-memory condition, and c) a fading-memory condition, these two related to the decoder $D$ (see Fig.~\ref{fig:diagramas}).  
It is worth noting that, as defined in~\cite{mardah05}, these assumptions upon $D$ exclude the use of side information by the decoder and/or the possibility of $D$ being affected by random noise or having a random internal state which is non-observable (please see~\cite{mardah05} for a detailed description of these assumptions). 

The inequality~\eqref{eq:la_de_Massey} has recently been extended in~\cite[Theorem~1]{li-eli11}, for the case of discrete-valued random variables and assuming $\rvas\Perp(\rvar,\rvap,\rvaq)$, as the following identity (written in terms of the signals and setup shown in Fig.~\ref{fig:diagramas}-(a)):
\begin{align}\label{eq:latest1}
 I(\rvax^{k}\to\rvay^{k})
 =
 I(\rvap^{k},\rvay^{k})
 +
 I(\rvax^{k}\to\rvay^{k}|\rvap^{k} ).
\end{align}
Letting $\rvaq=\rvas$ in Fig.~\ref{fig:diagramas}-(a) and with the additional assumption that $(\rvap,\rvas)\Perp\rvaq$, it was also shown in~\cite[Theorem~1]{li-eli11} that
\begin{align}\label{eq:latest2}
 I(\rvax^{k}\to\rvay^{k})
 =
 I(\rvap^{k};\rvay^{k})
 +
 I(\rvaq^{k-1};\rvay^{k})
+
 I(\rvap^{k};\rvaq^{k-1}|\rvay^{k}),
\end{align}
for the cases in which $\rvau(i)=\rvay(i)+\rvaq(i)$ (i.e., when the concatenation of $\Ssp_{4}$ and $\Ssp_{1}$ corresponds to a summing node).
In~\cite{li-eli11},~\eqref{eq:latest1} and~\eqref{eq:latest2} play important roles in characterizing the capacity of channels with noisy feedback.

To the best of our knowledge,%
~\eqref{eq:Martins_first},~\eqref{eq:la_de_martins08},~\eqref{eq:la_de_Massey}~\eqref{eq:mardah_dir_minus_mutual},~\eqref{eq:latest1} and~\eqref{eq:latest2}
are the only results available in the literature which lower bound the difference between an internal-to-internal directed information and an external-to-internal mutual information.
There exist even fewer published results in relation to inequalities between two directed informations involving only signals internal to the loop.
To the best of our knowledge, the only inequality of this type in the literature is the one found in the proof of 
Theorem~4.1 of~\cite{silder11}.
The latter takes the form of a (quasi) data-processing inequality for directed informations in closed-loop systems, and states that
\begin{align}\label{eq:silder11}
I(\rvax^{k}\to \rvay^{k}\parallel\rvaq^{k})\geq  I(\rvax^{k}\to \rvau^{k}),                                
\end{align}
provided%
\footnote{Here, and in the sequel, we use the notation $\rvax \Perp\rvay $ 
to mean ``$\rvax$ is independent of $\rvay$''.}
 $\rvaq\Perp(\rvar,\rvap)$
and if $\Ssp_{4}$ is such that $\rvay^{i}$ is a function of $(\rvau^{i},\rvaq^{i})$ (i.e., if $\Ssp_{4}$ is conditionally invertible) $\forall i$.
In~\eqref{eq:silder11}, 
\begin{align}
I(\rvax^{k}\to \rvay^{k}\parallel\rvaq^{k})
\eq 
\Sumfromto{i=1}{k}I(\rvay(i);\rvax^{i}|\rvay^{i-1},\rvaq^{i})
\end{align}
corresponds to the causally conditioned directed information defined in~\cite{kramer98}.   
Inequality~\eqref{eq:silder11} plays a crucial role~\cite{silder11}, since it allowed lower bounding the average data rate across a digital error-free channel by a directed information. (In~\cite{silder11}, $\rvaq$ corresponded to a random dither signal in an entropy-coded dithered quantizer.)

In this paper, we derive a set of information identities and inequalities involving pairs of sequences (internal or external to the loop) in feedback systems.
The first of these is an identity which, under an independence condition, can be interpreted as a law of conservation of information flows. 
The latter identity is the starting point for most of the results which follow it.
Among other things,
we extend~\eqref{eq:la_de_martins08} and~\eqref{eq:mardah_dir_minus_mutual} to the general setup depicted in Fig.~\ref{fig:diagramas}-(a), where 
\textit{none of the assumptions made in~\cite{mardah05,mardah08,silder11} (except causality) needs to hold}.
Moreover, we will prove the validity of~\eqref{eq:silder11} without assuming the 
conditional invertibility of $\Ssp_{4}$ nor that $\rvaq\Perp (\rvar,\rvap)$.
The latter result is one of four novel data-processing inequalities derived in Section~\ref{ssec:nested_directed}, each involving two nested directed informations valid for the system depicted in Fig.~\ref{fig:diagramas}-(a).
The last of these is a complete closed-loop counterpart of the traditional open-loop data-processing inequality.

The remainder of this paper begins with a description of the systems under study
and the extension of Massey's directed information
to the case in which each of the blocks in the loop may introduce an arbitrary, non-negative delay 
(i.e., we do not allow for anticipation).
The information identities and inequalities are presented in Section~\ref{sec:results}.
For clarity of the exposition, all the proofs are deferred to Section~\ref{sec:proofs}.
A brief discussion of potential applications of our results is presented in Section~\ref{sec:possible_applications}, which is followed by the conclusions in Section~\ref{sec:conclusions}.

\section{Preliminaries}
\subsection{System Description}
We begin by providing a formal description of the systems labeled $\Ssp_{1}\ldots \Ssp_{4}$ in Fig.~\ref{fig:diagramas}-(a). 
Their input-output relationships are given by the possibly-varying deterministic mappings%
\footnote{For notational simplicity, we omit writing their time dependency explicitly.}
\begin{subequations}\label{eq:block_defs}
 \begin{align}
 \rvae(i)&= \Ssp_{1}(\rvau^{i- d_{1}(i)},\rvar^{i}),
\\
\rvax(i)&= \Ssp_{2}(\rvae^{i- d_{2}(i)},\rvap^{i}),
\\
\rvay(i)&= \Ssp_{3}(\rvax^{i- d_{3}(i)},\rvas^{i}),
\\
\rvau(i)&= \Ssp_{4}(\rvay^{i- d_{4}(i)},\rvaq^{i}),
\end{align}
\end{subequations}
where $\rvar,\rvap,\rvas,\rvaq$ are exogenous random signals and the (possibly time-varying) delays $d_{1},d_{2},d_{3},d_{4}\in\set{0,1,\ldots}$ are such that
$$
d_{1}(k)
+
d_{2}(k)+
d_{3}(k)+
d_{4}(k)
\geq  1,\fspace \forall k\in\Nl.
$$
That is, the concatenation of $\Ssp_{1},\ldots,\Ssp_{4}$ has a delay of at least  one sample.
For every $i\in\set{1,\ldots,k}$, $\rvar(i)\in\Rl^{n_{\rvar}(i)}$, i.e., $\rvar(i)$ is a real random vector whose dimension is given by some function $n_{\rvar}:\set{1,\ldots,k}\to \Nl$.
The other sequences ($\rvaq,\rvap,\rvas,\rvax,\rvay,\rvau$) are defined likewise.

\subsection{A Necessary Modification of the Definition of Directed Information}
As stated in~\cite{massey90}, the directed information (as defined in~\eqref{eq:directed_inf_def}) is a more meaningful measure of the flow of information between $\rvax^{k}$ and $\rvay^{k}$ than the conventional mutual information $I(\rvax^{k};\rvay^{k})=\sumfromto{i=1}{k}I(\rvay(i);\rvax^{k}| \rvay^{i-1})$ when there exists causal feedback from $\rvay$ to $\rvax$.
In particular, if $\rvax^{k}$ and $\rvay^{k}$ are discrete-valued sequences, input and output, respectively, of a forward channel, and if there exists \textit{strictly causal}, perfect feedback, so that $\rvax(i)=\rvay(i-1)$  (a scenario utilized in~\cite{massey90} as part of an argument in favor of the directed information), then the mutual information becomes 
\begin{align*}
I(\rvax^{k};\rvay^{k})
&= 
H(\rvay^{k}) - H(\rvay^{k}|\rvax^{k})
=
H(\rvay^{k}) - H(\rvay^{k}|\rvay^{k-1})
=
H(\rvay^{k}) - H(\rvay(k)|\rvay^{k-1})
= 
H(\rvay^{k-1}).
\end{align*}
Thus, when strictly causal feedback is present, $I(\rvax^{k};\rvay^{k})$ fails to account for how much information about $\rvax^{k}$ has been conveyed to $\rvay^{k}$ through the forward channel that lies between them.

It is important to note that, in~\cite{massey90} (as well as in many works concerned with communications), the forward channel is instantaneous, i.e., it has no delay.
Therefore, if a feedback channel is utilized, then this feedback channel must have a delay of at least one sample, as in the example above.
However, when studying the system in Fig.~\ref{fig:diagramas}-(a), we may need to evaluate the directed information between signals $\rvax^{k}$ and $\rvay^{k}$ which are, respectively, input and output of a \textit{strictly casual} forward channel (i.e., with a delay of at least one sample), 
whose output is instantaneously fed back to its input.
In such case, if one further assumes perfect feedback and sets $\rvax(i)=\rvay(i)$, then, in the same spirit as before, 
\begin{align*}
I(\rvax^{k}\to \rvay^{k})
&= 
\Sumfromto{i=1}{k}I(\rvay(i);\rvax^{i}| \rvay^{i-1})
=
\Sumfromto{i=1}{k}\left[H(\rvay(i)|\rvay^{i-1}) - H(\rvay(i)|\rvax^{i}, \rvay^{i-1}) \right]
=
H(\rvay^{k}).  
\end{align*}  
As one can see, Massey's definition of directed information ceases to be meaningful if instantaneous feedback is utilized.

It is natural to solve this problem by recalling that, in the latter example, the forward channel had a delay, say $d$, greater than one sample.
Therefore, if we are interested in measuring how much of the information in $\rvay(k)$, not present in $\rvay^{i-1}$, was conveyed from $\rvax^{i}$ through the forward channel, we should look at the mutual information $I(\rvay(i);\rvax^{i-d}|\rvay^{i-1})$, because only the input samples $\rvax^{i-d}$ can have an influence on $\rvay(i)$.
For this reason, 
we introduce
the following, modified notion of directed information
\begin{defn}[Directed Information with Forward Delay]
\textit{In this paper, the directed information from $\rvax^{k}$ to $\rvay^{k}$ through a forward channel with a non-negative time varying delay of $d(i)$ samples is defined as
\begin{align}
 I(\rvax^{k}\to \rvay^{k}) \eq \Sumfromto{i=1}{k}I(\rvay(i);\rvax^{i-d(i)}|\rvay^{i-1}).
\end{align}}
\end{defn}
For a zero-delay forward channel, the latter definition coincides with Massey's.

Likewise, we adapt the definition of causally-conditioned directed information to the definition
\begin{align*}
 I(\rvax^{k}\to \rvay^{k}\parallel\rvae^{k} ) 
 \eq 
 \Sumfromto{i=1}{k}I(\rvay(i);\rvax^{i-d_{3}(i)}|\rvay^{i-1},\rvae^{i-d_{2}(i)}).
\end{align*}
when the signals $\rvae$, $\rvax$ and $\rvay$ are related according to~\eqref{eq:block_defs}.

Before finishing this section, it is convenient to recall
the following identity (a particular case of the chain rule of conditional mutual information~\cite{yeung-02}), which will be extensively utilized in the proofs of our results:
\begin{align}\label{eq:chainrule_I}
 I(\rvaa,\rvab;\rvac|\rvad) = I(\rvab;\rvac|\rvad) + I(\rvaa;\rvac|\rvab,\rvad).
\end{align}

\section{Information Identities and Inequalities}\label{sec:results}
\subsection{Relationships Between Mutual and Directed Informations}
We begin by stating a fundamental result, \textit{which relates the directed information between two signals within a feedback loop, say $\rvax$ and $\rvay$, to the mutual information between an external set of signals and $\rvay$}:

\begin{thm}\label{thm:main}
\textit{ In the system shown in Fig.~\ref{fig:diagramas}-(a), it holds that
\begin{align}\label{eq:main_thm}
 I(\rvax^{k}\to \rvay^{k}) 
=
 I(\rvaq^{k},\rvar^{k},\rvap^{k}\to\rvay^{k})
- 
I(\rvaq^{k},\rvar^{k},\rvap^{k}\to\rvay^{k}\parallel \rvax^{k})
\leq I(\rvap^{k},\rvaq^{k},\rvar^{k}\,;\rvay^{k}), \fspace \forall k\in\Nl,
\end{align}
with equality achieved if $\rvas$ is independent of $(\rvap,\rvaq,\rvar)$. }
\finenunciado
\end{thm}
This fundamental result, which for the cases in which $\rvas\Perp(\rvap,\rvaq,\rvar)$ can be understood as a \textit{law of conservation of information flow}, is illustrated in Fig.~\ref{fig:information_flow}.
For such cases, the information causally conveyed from $\rvax$ to $\rvay$ equals the information flow from $(\rvaq,\rvar,\rvap)$ to $\rvay$.   
When $(\rvap,\rvaq,\rvar)$ are not independent of $\rvas$, part of the mutual information between $(\rvap,\rvaq,\rvar)$ and $\rvay$ 
(corresponding to the term $I(\rvaq^{k},\rvar^{k},\rvap^{k}\to\rvay^{k}\parallel \rvax^{k})$) 
can be thought of as being ``leaked'' through $\rvas$, thus bypassing the forward link from $\rvax$ to $\rvay$.
This provides an intuitive interpretation for~\eqref{eq:main_thm}. 
\begin{figure}[htpb]
 \centering
\input{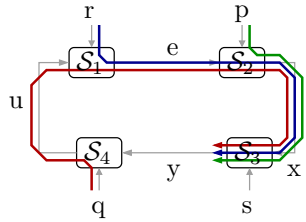}
\caption{The flow of information between exogenous signals $(\rvap,\rvaq,\rvar)$ and the internal signal $\rvay$ equals the directed information from $\rvax^{k}$ to $\rvay^{k}$ when $\rvas\Perp(\rvap,\rvaq,\rvar)$.}
\label{fig:information_flow}
\end{figure}

\begin{rem}
Theorem~\ref{thm:main} implies that $I(\rvax^{k}\to \rvay^{k})$ is only a part of (or at most equal to) the  information ``flow'' between all the exogenous signals entering the loop outside the link $\rvax\to \rvay$ 
(namely $(\rvaq,\rvar,\rvap)$), and $\rvay$.
In particular, if $(\rvap,\rvaq,\rvar)$ were deterministic, then $I(\rvax^{k}\to \rvay^{k})=0$, regardless of the blocks $\Ssp_{1},\ldots,\Ssp_{4}$ and irrespective of the nature of $\rvas$. 
\finenunciado
\end{rem}

\begin{rem}
By using~\eqref{eq:chainrule_I}, 
$
I(\rvap^{k},\rvaq^{k},\rvar^{k};\rvay^{k})
=
I(\rvar^{k};\rvay^{k})
+
I(\rvap^{k},\rvaq^{k};\rvay^{k}|\rvar^{k})
$.
Then, applying 
Theorem~\ref{thm:main}, we recover~\eqref{eq:la_de_Massey}, whenever $\rvas\Perp(\rvaq,\rvar,\rvap)$.
Thus,~\cite[Theorem~3]{massey90} and~\cite[Lemma~4.8.1]{tatiko00}) can be obtained as a corollary of Theorem~\ref{thm:main}.
\finenunciado
\end{rem}

The following result provides an inequality relating $I(\rvax^{k}\to \rvay^{k})$ with the separate flows of information
$I(\rvar^{k} ; \rvay^{k})$ and $I(\rvap^{k},\rvaq^{k}\,;\,\rvay^{k})$.
\begin{thm}\label{thm:from_splitting_more_precise}
\textit{For the system shown in Fig.~\ref{fig:diagramas}-(a), if 
$\rvas\Perp (\rvap,\rvaq,\rvar)$ and 
$\rvar^{k}\Perp (\rvap^{k},\rvaq^{k})$, then 
\begin{align}
I(\rvax^{k}\to \rvay^{k}) 
&
\geq 
I(\rvar^{k} ; \rvay^{k})+I(\rvap^{k},\rvaq^{k}\,;\,\rvay^{k}). \label{eq:logro}
\end{align}
with equality if and only if the Markov chain $(\rvap^{k},\rvaq^{k})\leftrightarrow\rvay^{k}\leftrightarrow\rvar^{k}$ holds.}
\end{thm}
Theorem~\ref{thm:from_splitting_more_precise} 
shows that, provided $(\rvap,\rvaq,\rvar)\Perp \rvas$, 
$I(\rvax^{k}\to \rvay^{k})$ is lower bounded by the sum of the individual flows from all the subsets in any given partition of $(\rvap^{k},\rvaq^{k},\rvar^{k})$, to $\rvay^{k}$, provided these subsets are mutually independent. 
Indeed,
both theorems~\ref{thm:main} and~\ref{thm:from_splitting_more_precise} can be generalized for any appropriate choice of external and internal signals.
More precisely, let $\Theta$ be the set of all external signals in a feedback system.
Let $\alpha$ and $\beta$ be two internal signals in the loop.
Define $\Theta_{\alpha,\beta}\subset \Theta$ as the set of exogenous signals which are introduced to the loop 
at every subsystem $\Ssp_{i}$ that lies in the path going from $\alpha$ to $\beta$. 
Thus, for any $\rho\in\Theta \setminus \Theta_{\alpha,\beta}$, if $\Theta_{\alpha,\beta} \Perp \Theta\setminus\Theta_{\alpha,\beta}$, we have that~\eqref{eq:main_thm} and~\eqref{eq:logro} become 
\begin{align}
I(\alpha\to\beta)&=
I(\Theta\setminus\set{\Theta_{\alpha,\beta}};\beta),
\\
 I(\alpha\to\beta) - I(\rho;\beta) &\geq I(\Theta\setminus\set{\rho\cup\Theta_{\alpha,\beta}};\beta),
\end{align}
respectively.

To finish this section, we present a stronger, non-asymptotic version of inequality~\eqref{eq:mardah_dir_minus_mutual}:
\begin{thm}\label{thm:three_full_loops}
 \textit{In the system shown in Fig.~\ref{fig:diagramas}-(a), if $(\rvar,\rvap,\rvaq,\rvas)$ are mutually independent, then}
\begin{align}\label{eq:nice}
 I(\rvax^{k}\to\rvay^{k})
&=
I(\rvar^{k};\rvau^{k}) 
+ 
I(\rvap^{k};\rvae^{k})
+ 
I(\rvaq^{k};\rvay^{k}) 
+
I(\rvap^{k};\rvau^{k}|\rvae^{k})
+ 
I(\rvar^{k},\rvap^{k};\rvay^{k}|\rvau^{k}).
\end{align}
\finenunciado
\end{thm}
As anticipated, Theorem~\ref{thm:three_full_loops} can be seen as an extension of~\eqref{eq:mardah_dir_minus_mutual} 
 to the more general setup shown in Fig.~\ref{fig:diagramas}-(a), where the assumptions made 
in~\cite[Lemma~4.1]{mardah05} do not need to hold. 
In particular, letting the decoder $D$ and $\rvax_{0}$ in Fig.~\ref{fig:diagramas}-(b) correspond to $\Ssp_{4}$ and $\rvap^{k}$ in Fig.~\ref{fig:diagramas}-(a), respectively, we see that inequality~\eqref{eq:logro} holds even if $D$ and $E$ have dependent initial states, or if the internal state of $D$ is not observable~\cite{gogrsa01}.

Theorem~\ref{thm:three_full_loops} also admits an interpretation in terms of information flows.
This can be appreciated in the diagram shown in Fig.~\ref{fig:flujos2}, which 
depicts the individual full-turn flows (around the entire feedback loop) stemming from $\rvaq$, $\rvar$ and $\rvap$.
Theorem~\ref{thm:three_full_loops} states that the sum of these individual flows is a lower bound for the 
directed information from $\rvax$ to $\rvay$, provided $\rvaq,\rvar,\rvap,\rvas$ are independent.
\begin{figure}[htbp]
 \centering
 \input{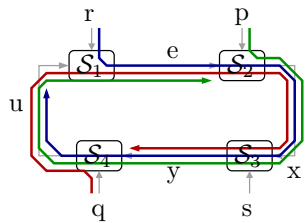}
 \caption{A representation of the three first information flows on the right-hand-side of~\eqref{eq:nice}.}
 \label{fig:flujos2}
\end{figure}

\subsection{Relationships Between Nested Directed Informations}\label{ssec:nested_directed}
This section presents three closed-loop versions of the data processing inequality \textit{relating two directed informations}, both between pairs of signals \textit{internal} to the loop.
As already mentioned in Section~\ref{sec:intro},
to the best of our knowledge, the first inequality of this type to appear in the literature is the one in
Theorem~4.1 in~\cite{silder11} (see~\eqref{eq:silder11}).
Recall that the latter result stated that 
$I(\rvax^{k}\to\rvay^{k}\parallel \rvaq^{k})\geq I(\rvax^{k}\to\rvau^{k})$,
requiring
$\Ssp_{4}$ to be such that $\rvay^{i}$ is a deterministic function of $(\rvau^{i},\rvaq^{i})$ and that $\rvaq\Perp(\rvar,\rvap)$.
The following result presents another inequality which also relates two nested directed informations, namely, 
$
I(\rvax^{k}\to \rvay^{k})$ and $I(\rvae^{k}\to \rvay^{k})
$,  
but requiring only that $\rvas\Perp (\rvaq,\rvar,\rvap)$. 
\begin{thm}\label{thm:DPI_dir_dir}
\textit{For the closed-loop system in Fig.~\ref{fig:diagramas}-(b), if 
$(\rvaq,\rvar,\rvap)\Perp \rvas$, then
 \begin{align}
  I(\rvax^{k}\to \rvay^{k}) &\geq I(\rvae^{k}\to \rvay^{k}).
 \end{align}}
 \finenunciado
\end{thm}
Notice that Theorem~\ref{thm:DPI_dir_dir} does not require $\rvap$ to be independent of $\rvar$ or $\rvaq$.
This may seem counter-intuitive upon noting that $\rvap$ enters the loop between the link from $\rvae$ to $\rvax$. 

The following theorem is an identity between two directed informations involving only internal signals.
It can also be seen as a  complement to Theorem~\ref{thm:DPI_dir_dir}, since it can be directly applied to establish the relationship
between $ I(\rvae^{k}\to \rvay^{k})$
and $I(\rvae^{k}\to \rvau^{k})$.

\begin{thm}\label{thm:finally}
\textit{For the system shown in Fig.~\ref{fig:diagramas}-(a), 
if $(\rvaq,\rvas)\Perp (\rvar,\rvap)$, then
\begin{align}\label{eq:finally0}
 I(\rvax^{k}\to \rvay^{k})
\leq 
I(\rvax^{k}\to \rvau^{k})
+
I(\rvaq^{k}\,;\, \rvay^{k})
+
I(\rvar^{k},\rvap^{k}; \rvay^{k}| \rvau^{k})
+ I(\rvaq^{k};\rvar^{k}|\rvau^{k},\rvay^{k}).
\end{align}
with equality if, in addition, $\rvaq\Perp\rvas$. 
In the latter case, it holds that
\begin{align}\label{eq:finally}
I(\rvax^{k}\to \rvay^{k})
=
I(\rvax^{k}\to \rvau^{k})
+
I(\rvaq^{k}\,;\, \rvay^{k})
+
I(\rvar^{k},\rvap^{k}; \rvay^{k}| \rvau^{k}).
\end{align}
}\finenunciado
\end{thm}
Notice that, by requiring additional independence conditions upon the exogenous signals (specifically, $\rvaq\Perp\rvas$), Theorem~\ref{thm:finally} (and, in particular,~\eqref{eq:finally})  yields 
\begin{align}\label{eq:la_algo_mejor}
I(\rvax^{k}\to \rvay^{k})
\geq
I(\rvax^{k}\to \rvau^{k}),
\end{align}
which strengthens 
the inequality in~\cite[Theorem~4.1]{silder11} (stated above in~\eqref{eq:silder11}).
More precisely,~\eqref{eq:la_algo_mejor} does not require conditioning one of the directed informations and holds irrespective of the invertibility of the mappings in the loop.

A closer counterpart of~\eqref{eq:silder11} (i.e., of~\cite[Theorem~4.1]{silder11}), involving 
$ I(\rvax^{k}\to\rvay^{k}\parallel \rvaq^{k})$, is presented next.
\begin{thm}\label{thm:xtoycond}
\textit{For the system shown in Fig.~\ref{fig:diagramas}-(a), 
if 
$(\rvaq,\rvas)\Perp (\rvar,\rvap)$, then 
\begin{align}\label{eq:xtoyxtoucond}
 I(\rvax^{k}\to\rvay^{k}| \rvaq^{k})
=
I(\rvax^{k}\to\rvau^{k})
+
I(\rvar^{k},\rvap^{k}; \rvay^{k}| \rvau^{k})
 + I(\rvaq^{k};\rvar^{k}|\rvau^{k},\rvay^{k})
\overset{(\dagger)}{=}
I(\rvax^{k}\to\rvay^{k}\parallel \rvaq^{k}).
\end{align}}
\textit{
where the equality labeled $(\dagger)$ hods if, in addition,  
the Markov chain
\begin{align}\label{eq:MC_q_and_s}
\rvaq_{i+1}^{k}
\longleftrightarrow
\rvaq^{i}
\longleftrightarrow
\rvas^{i}
\end{align}
is satisfied for all $i\in\set{1,\ldots,k}$. 
}
\finenunciado
\end{thm}
Thus, provided $(\rvaq,\rvas)\Perp (\rvar,\rvap)$,~\eqref{eq:xtoyxtoucond} yields that~\eqref{eq:silder11} holds regardless of the invertibility of $\Ssp_{4}$, requiring instead that, for all $i\in\set{1,\ldots,k}$, any statistical dependence between $\rvaq^{k}$ and $\rvas^{i}$ resides only in $\rvaq^{i}$ (i.e., that Markov chain~\eqref{eq:MC_q_and_s} holds).

The results derived so far relate directed informations having either the same ``starting'' sequence or the same ``destination'' sequence. 
We finish this section with 
the following corollary, which follows directly by combining theorems~\ref{thm:DPI_dir_dir} and~\ref{thm:finally} and relates directed informations involving four different sequences internal to the loop.
\begin{coro}[Full Closed-Loop Directed Data Processing Inequality]\label{coro:full_D-DPI}
\textit{  For the system shown in Fig.~\ref{fig:diagramas}-(a), if 
$(\rvaq,\rvas)\Perp (\rvar,\rvap)$
and
$\rvaq\Perp\rvas$,
then
\begin{align}
 I(\rvax^{k}\to \rvay^{k})
\overset{(a)}{\geq} 
I(\rvae^{k}\to \rvau^{k}) 
+
I(\rvaq^{k}\,;\, \rvay^{k})
+
I(\rvar^{k}; \rvay^{k}| \rvau^{k})
\geq 
I(\rvae^{k}\to \rvau^{k}).
\end{align}
Equality holds in $(a)$ if, in addition, $\rvar\Perp\rvap$ (i.e., if $(\rvaq,\rvar,\rvap,\rvas)$ are mutually independent).}
\finenunciado
\end{coro}
To the best of our knowledge, Corollary~\ref{coro:full_D-DPI} is the first result available in the literature 
providing a lower bound to the gap between two nested directed informations, involving four different signals inside the feedback loop.
This result can be seen as the first full extension of the open-loop (traditional) data-processing inequality, to  arbitrary closed-loop scenarios. 
(Notice that there is no need to consider systems with more than four mappings, since all external signals entering the loop between a given pair of internal signals can be regarded as exogenous inputs to a single equivalent deterministic mapping.)

\section{Proofs}\label{sec:proofs}
We start with the proof of Theorem~\ref{thm:main}.
\begin{proof}[Proof of Theorem~\ref{thm:main}]
It is clear from Fig.~\ref{fig:diagramas}-(a) and from~\eqref{eq:block_defs}
that the relationship between $\rvar$, $\rvap$, $\rvaq$, $\rvas$, $\rvax$ and $\rvay$ can be represented by the diagram shown in Fig.~\ref{fig:Key}. 
\begin{figure}[h]
\centering
\input{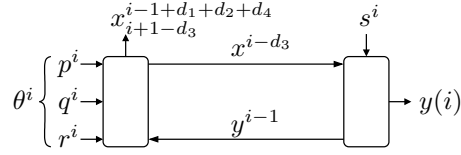}
\caption{Representation of the system of Fig.~\ref{fig:diagramas}-(b) highlighting the dependency between $p$, $q$, $r$, $s$, $x$ and $y$.
The dependency on $i$ of the delays $d_{1}(i),\ldots, d_{4}(i)$ is omitted for clarity.}
\label{fig:Key}
\end{figure}
From this diagram and Lemma~\ref{lem:not_so_obvious} (in the appendix) it follows that
if $\rvas$ is independent of $(\rvar,\rvap,\rvaq)$, then the following Markov chain holds:
\begin{align}
 \rvay(i) &
\longleftrightarrow 
(\rvax^{i-d_{3}(i)},\rvay^{i-1})
\longleftrightarrow 
(\rvap^{i},\rvaq^{i},\rvar^{i}).\label{eq:MC1}
\end{align}
Denoting the triad of exogenous signals $\rvap^{k},\rvaq^{k},\rvar^{k}$ by 
\begin{align}
 \theta^{k} \eq (\rvap^{k},\rvaq^{k},\rvar^{k}),
\end{align}
we have the following

\begin{subequations}\label{eq:lanueva}
\begin{align}
 I(\rvax^{k}\to \rvay^{i})
&=
\Sumfromto{i=1}{k}I(\rvay(i);\rvax^{i-d_{3}(i)}|\rvay^{i-1})
\nonumber\\& 
\overset{\eqref{eq:chainrule_I}}{=}
\Sumfromto{i=1}{k}
\left[
I(\theta^{i},\rvax^{i-d_{3}(i)};\rvay(i)|\rvay^{i-1})
-
I(\theta^{i};\rvay(i)|\rvax^{i-d_{3}(i)},\rvay^{i-1})
\right]
\nonumber\\& 
\overset{(a)}{=}
\Sumfromto{i=1}{k}
\left[
I(\theta^{i};\rvay(i)|\rvay^{i-1})
-
I(\theta^{i};\rvay(i)|\rvax^{i-d_{3}(i)},\rvay^{i-1})
\right]\label{eq:solo_directeds}
\\& 
\overset{(b)}{\leq}
\Sumfromto{i=1}{k}
I(\theta^{i};\rvay(i)|\rvay^{i-1})
\overset{(c)}{\leq}
\Sumfromto{i=1}{k}
I(\theta^{k};\rvay(i)|\rvay^{i-1})
\\&=
I(\theta^{k};\rvay^{k}).
\end{align}
\end{subequations}
In the above,
 $(a)$
follows from the fact that, if $\rvay^{i-1}$ is known, then $\rvax^{i-d_{3}(i)}$ is a deterministic function of $\theta^{i}$.
The resulting sums on the right-hand side of~\eqref{eq:solo_directeds} correspond to
$
 I(\rvaq^{k},\rvar^{k},\rvap^{k}\to\rvay^{k})
- 
I(\rvaq^{k},\rvar^{k},\rvap^{k}\to\rvay^{k}\parallel \rvax^{k})
$,
and thereby proving the first part of the theorem, i.e., the equality in~\eqref{eq:main_thm}.
In turn, $(b)$ stems from the non-negativity of mutual informations, turning into equality if $\rvas\Perp(\rvar,\rvap,\rvaq)$, as a direct consequence of the Markov chain in~\eqref{eq:MC1}.
Finally, equality holds in~$(c)$ if $\rvas\Perp(\rvaq,\rvar,\rvap)$, since
$\rvay$ depends causally upon $\theta$.
This shows that equality in~\eqref{eq:main_thm} is achieved if $\rvas\Perp(\rvaq,\rvar,\rvap)$, completing  the proof.
\end{proof}

\begin{proof}[Proof of Theorem~\ref{thm:from_splitting_more_precise}]
Apply the chain-rule identity~\eqref{eq:chainrule_I} to the RHS of~\eqref{eq:main_thm} to obtain
\begin{align}\label{eq:Itheta_to_sum}
  I(\theta^{k};\rvay^{k})
=
  I(\rvap^{k},\rvaq^{k},\rvar^{k};\rvay^{k})
=
  I(\rvap^{k},\rvaq^{k};\rvay^{k}|\rvar^{k})
+
I(\rvar^{k};\rvay^{k}).
\end{align}
Now, applying~\eqref{eq:chainrule_I} twice, one can express the term  
$I(\rvap^{k},\rvaq^{k};\rvay^{k}|\rvar^{k})$
as follows:
\begin{equation}\label{eq:estaotra}
 \begin{split}
 I(\rvap^{k},\rvaq^{k};\rvay^{k}|\rvar^{k})
&=
I(\rvap^{k},\rvaq^{k}\,;\,\rvay^{k},\rvar^{k})
-
I(\rvap^{k},\rvaq^{k};\rvar^{k})
=
I(\rvap^{k},\rvaq^{k}\,;\,\rvay^{k},\rvar^{k})
\\&
=
I(\rvap^{k},\rvaq^{k};\rvay^{k})
+
I(\rvap^{k},\rvaq^{k};\rvar^{k}|\rvay^{k}),
\end{split}
\end{equation}
where the second equality follows since $(\rvap^{k},\rvaq^{k})\Perp \rvar^{k}$. 
The result then follows directly by combining~\eqref{eq:estaotra} with~\eqref{eq:Itheta_to_sum} and~\eqref{eq:main_thm}.
\end{proof}

\begin{proof}[Proof of Theorem~\ref{thm:three_full_loops}]
Since $\rvaq\Perp(\rvar,\rvap,\rvas)$, 
 \begin{align}
 I(\rvax^{k}\to\rvay^{k})
&
\overset{(a)}{=}
I(\rvax^{k}\to\rvau^{k}) + I(\rvaq^{k};\rvay^{k}) + I(\rvar^{k},\rvap^{k};\rvay^{k}|\rvau^{k})
\\
&
\overset{(b)}{=}
I(\rvar^{k},\rvap^{k};\rvau^{k}) + I(\rvaq^{k};\rvay^{k}) + I(\rvar^{k},\rvap^{k};\rvay^{k}|\rvau^{k})
\\
&
\overset{(c)}{=}
I(\rvar^{k};\rvau^{k}) + I(\rvap^{k};\rvau^{k}|\rvar^{k}) + I(\rvaq^{k};\rvay^{k}) + I(\rvar^{k},\rvap^{k};\rvay^{k}|\rvau^{k}),
\label{eq:laotraultima}
\end{align}
where $(a)$ is due to Theorem~\ref{thm:finally}, $(b)$ follows from Theorem~\ref{thm:main} and the fact that $(\rvas,\rvaq)\Perp (\rvar,\rvap)$ and $(c)$ from the chain rule of mutual information.
For the second term on the RHS of the last equation, we have 
\begin{align}
 I(\rvap^{k};\rvau^{k}|\rvar^{k})
&
\overset{(a)}{=}
I(\rvap^{k};\rvau^{k}|\rvar^{k}) + I(\rvap^{k};\rvar^{k})
=
I(\rvap^{k};\rvar^{k},\rvau^{k})
\\&\overset{(b)}{=}
I(\rvap^{k};\rvar^{k},\rvau^{k},\rvae^{k})
-
I(\rvap^{k};\rvae^{k}|\rvar^{k},\rvau^{k})
\\&\overset{(c)}{=}
I(\rvap^{k};\rvar^{k},\rvau^{k},\rvae^{k})
\\&\overset{(d)}{=}
I(\rvap^{k};\rvae^{k})
+
I(\rvap^{k};\rvar^{k},\rvau^{k}|\rvae^{k})
\\&\overset{(e)}{=}
I(\rvap^{k};\rvae^{k})
+
I(\rvap^{k};\rvau^{k}|\rvae^{k})
+
I(\rvap^{k};\rvar^{k}|\rvau^{k},\rvae^{k})
\\&\overset{(f)}{=}
I(\rvap^{k};\rvae^{k})
+
I(\rvap^{k};\rvau^{k}|\rvae^{k}),
\label{eq:ultima_linea}
\end{align}
where $(a)$ holds since $\rvar\Perp\rvap$, $(b)$, $(d)$ and $(e)$ stem from the chain rule of mutual information~\eqref{eq:chainrule_I}, and $(c)$ is a consequence of the Markov chain 
$\rvae^{k} \leftrightarrow
(\rvau^{k},\rvar^{k})
\leftrightarrow
\rvap^{k}$ which is due to the fact that $\rvae^{k}=\Ssp_{1}(\rvau^{k-d_{1}(k)},\rvar^{k})$.
Finally, $(f)$ is due to the Markov chain 
$
\rvar^{k}
\leftrightarrow
(\rvau^{k},\rvae^{k})
\leftrightarrow
\rvap^{k}
$,
which holds because $\rvar\Perp(\rvap,\rvas,\rvaq)$ as a consequence of Lemma~\ref{lem:not_so_obvious} in the appendix (see also Fig.~\ref{fig:diagramas}-(a)).
Substitution of~\eqref{eq:ultima_linea} into~\eqref{eq:laotraultima} yields~\eqref{eq:nice}, thereby completing the proof.
\end{proof}

\begin{proof}[Proof of Theorem~\ref{thm:DPI_dir_dir}]
Since $(\rvap,\rvaq,\rvar)\Perp \rvas$, we can apply~\eqref{eq:la_de_Massey} (where now $(\rvaq,\rvar)$ plays 
the role of $\rvar$), 
and obtain
\begin{align}
 I(\rvax^{k}\to \rvay^{k})\geq I(\rvaq^{k},\rvar^{k};\rvay^{k}).
\end{align}
Now, we apply Theorem~\ref{thm:main}, which gives  
\begin{align}
 I(\rvaq^{k},\rvar^{k};\rvay^{k}) \geq  I(\rvae^{k}\to \rvay^{k}),
\end{align}
completing the proof.
\end{proof}

\begin{proof}[Proof of Theorem~\ref{thm:finally}]
Applying Theorem~\ref{thm:main}, since $(\rvar,\rvap) \Perp(\rvas,\rvaq)$, 
\begin{align}\label{eq:qtoy=r,u}
 I(\rvax^{k}\to \rvau^{k}) = I(\rvar^{k},\rvap^{k}\,;\, \rvau^{k}).
\end{align}
For the other directed information, we have that  
\begin{align}
 I(\rvax^{k}\to \rvay^{k})
&
\overset{(a)}{\leq}
I(\rvar^{k},\rvap^{k},\rvaq^{k}\,;\, \rvay^{k})
\nonumber
\\&
\overset{\eqref{eq:chainrule_I}}{=}
I(\rvaq^{k}\,;\, \rvay^{k})
+
I(\rvar^{k},\rvap^{k}\,;\, \rvay^{k}| \rvaq^{k})
\label{eq:lamisma_I}
\\&
\overset{\eqref{eq:chainrule_I}}{=}
I(\rvaq^{k}\,;\, \rvay^{k})
+
I(\rvar^{k},\rvap^{k}\,;\, \rvau^{k},\rvay^{k}| \rvaq^{k})
-
I(\rvar^{k},\rvap^{k};\rvau^{k}|\rvaq^{k},\rvay^{k})
\nonumber
\\&
\overset{(b)}{=}
I(\rvaq^{k}\,;\, \rvay^{k})
+
I(\rvar^{k},\rvap^{k}\,;\, \rvau^{k},\rvay^{k}| \rvaq^{k})
\nonumber
\\&
\overset{\eqref{eq:chainrule_I}}{=}
I(\rvaq^{k}\,;\, \rvay^{k})
+
I(\rvar^{k},\rvap^{k}\,;\, \rvau^{k},\rvay^{k},\rvaq^{k})
-
I(\rvar^{k},\rvap^{k};\rvaq^{k})
\nonumber
\\&
\overset{\eqref{eq:chainrule_I}}{=}
I(\rvaq^{k}\,;\, \rvay^{k})
+
I(\rvar^{k},\rvap^{k}\,;\, \rvau^{k})
+
I(\rvar^{k},\rvap^{k}; \rvay^{k},\rvaq^{k} | \rvau^{k})
-
I(\rvar^{k},\rvap^{k};\rvaq^{k})
\nonumber
\\&
\overset{(c)}{\leq}
I(\rvaq^{k}\,;\, \rvay^{k})
+
I(\rvar^{k},\rvap^{k}\,;\, \rvau^{k})
+
I(\rvar^{k},\rvap^{k}; \rvay^{k},\rvaq^{k} | \rvau^{k})
\nonumber
\\&
\overset{\eqref{eq:chainrule_I}}{=}
I(\rvaq^{k}\,;\, \rvay^{k})
+
I(\rvar^{k},\rvap^{k}\,;\, \rvau^{k})
+
I(\rvar^{k},\rvap^{k}; \rvay^{k}| \rvau^{k})
 + I(\rvaq^{k};\rvar^{k}|\rvau^{k},\rvay^{k})
\label{eq:penultima}
\\&
\overset{(d)}{\leq}
I(\rvaq^{k}\,;\, \rvay^{k})
+
I(\rvar^{k},\rvap^{k}\,;\, \rvau^{k})
+
I(\rvar^{k},\rvap^{k}; \rvay^{k}| \rvau^{k}),\label{eq:ladeabajo}
\end{align}
where 
$(a)$ follows from Theorem~\ref{thm:main}, which also states that equality is reached if and only if $(\rvar,\rvap,\rvaq) \Perp \rvas$.
In turn, 
$(b)$ is due to the fact that 
$\rvau^{k}$ is a deterministic function of $\rvaq^{k},\rvay^{k}$.
Equality $(c)$ holds if and only if $(r,\rvap)\Perp \rvaq$.
Finally,
from Lemma~\ref{lem:not_so_obvious} (in the appendix),
$(d)$ turns into equality if $\rvaq\Perp (\rvar,\rvap,\rvas)$. 
Substitution of~\eqref{eq:qtoy=r,u} into~\eqref{eq:ladeabajo} yields~\eqref{eq:finally}, completing the proof.
\end{proof}

\begin{proof}[Proof of Theorem~\ref{thm:xtoycond}]
We begin with the second part of the theorem, proving the validity of the equality $(\dagger)$ in~\eqref{eq:xtoyxtoucond}.
We have the following:
{\allowdisplaybreaks
\begin{align}
I(\rvax^{k}\to\rvay^{k}\parallel \rvaq^{k})
&
=
\Sumfromto{i=1}{k}
I(\rvay(i);\rvax^{i-d_{3}(i)}|\rvay^{i-1},\rvaq^{i})
\\&
 \overset{\eqref{eq:chainrule_I}}{=}
\Sumfromto{i=1}{k}
\left[
I(\rvar^{i},\rvap^{i},\rvax^{i-d_{3}(i)};\rvay(i)|\rvay^{i-1},\rvaq^{i})
-
I(\rvar^{i},\rvap^{i};\rvay(i)|\rvax^{i-d_{3}(i)},\rvay^{i-1},\rvaq^{i})
\right]
\\&
 \overset{(a)}{\leq}
\Sumfromto{i=1}{k}
I(\rvar^{i},\rvap^{i},\rvax^{i-d_{3}(i)};\rvay(i)|\rvay^{i-1},\rvaq^{i})
\\&
 \overset{(b)}{=}
\Sumfromto{i=1}{k}
I(\rvar^{i},\rvap^{i};\rvay(i)|\rvay^{i-1},\rvaq^{i})
\\&
 \overset{\eqref{eq:chainrule_I}}{=}
\Sumfromto{i=1}{k}
\left[
I(\rvar^{i},\rvap^{i},\rvaq_{i+1}^{k};\rvay(i)|\rvay^{i-1},\rvaq^{i})
-
I(\rvaq_{i+1}^{k};\rvay(i)|\rvay^{i-1},\rvaq^{i},\rvar^{i},\rvap^{i})
\right]
\\&
 \overset{(c)}{=}
\Sumfromto{i=1}{k}
\left[
I(\rvar^{i},\rvap^{i},\rvaq_{i+1}^{k};\rvay(i)|\rvay^{i-1},\rvaq^{i})
\right]
\\&
 \overset{\eqref{eq:chainrule_I}}{=}
\Sumfromto{i=1}{k}
\left[
I(\rvar^{i},\rvap^{i};\rvay(i)|\rvay^{i-1},\rvaq^{k})
+
I(\rvaq_{i+1}^{k};\rvay(i)|\rvay^{i-1},\rvaq^{i})
\right]
\\&
 \overset{(d)}{=}
\Sumfromto{i=1}{k}
I(\rvar^{i},\rvap^{i};\rvay(i)|\rvay^{i-1},\rvaq^{k})
\label{eq:directed_normal_cond}
\\&
 \overset{(e)}{\leq}
\Sumfromto{i=1}{k}
I(\rvar^{k},\rvap^{k};\rvay(i)|\rvay^{i-1},\rvaq^{k})
=
I(\rvar^{k},\rvap^{k};\rvay^{i}|\rvaq^{k})
\label{eq:Irprgivenq}
\end{align}
}
where equality holds in $(a)$ if and only if 
the Markov chain 
$
\rvas^{i}\leftrightarrow 
\rvaq^{i}\leftrightarrow  
(\rvar^{i},\rvap^{i})
$
holds for all $i\in\set{1,\ldots, k}$ (as a straightforward extension of Lemma~\ref{lem:not_so_obvious}).
In our case, the latter Markov chain holds since we are assuming $(\rvaq^{k},\rvas^{k})\Perp (\rvar^{k},\rvap^{k})$.
In turn, $(b)$ stems from the fact that, for all $i\in\set{1,\ldots, k}$,  
$\rvax^{i-d_{3}(i)}$ is a function of 
$\rvay^{i-1},\rvaq^{i},\rvar^{i},\rvap^{i}$.
To prove~$(c)$, we resort to~\eqref{eq:chainrule_I} and write 
\begin{align}\label{eq:I=0+0_first}
 I(\rvaq_{i+1}^{k};\rvay(i)|\rvay^{i-1},\rvaq^{i},\rvar^{i},\rvap^{i})
&=
 I(\rvaq_{i+1}^{k};\rvay^{i},\rvar^{i},\rvap^{i}|\rvaq^{i})
-
 I(\rvaq_{i+1}^{k};\rvay^{i-1},\rvar^{i},\rvap^{i}|\rvaq^{i})
\end{align}
From the definitions of the blocks (in~\eqref{eq:block_defs}), it can be seen that, given $\rvaq^{i}$, the triad of random sequences $(\rvay^{i},\rvar^{i},\rvap^{i})$ is a deterministic function of 
(at most)
$(\rvas^{i},\rvar^{i},\rvap^{i})$. 
Recalling that 
$(\rvaq^{k},\rvas^{k})\Perp (\rvar^{},\rvap^{k})$ and that 
$\rvaq_{i+1}^{k}\leftrightarrow\rvaq^{i}\leftrightarrow \rvas^{i}$ (see~\eqref{eq:MC_q_and_s}), it readily follows that
$
\rvaq_{i+1}^{k}\leftrightarrow\rvaq^{i}\leftrightarrow (\rvar^{i},\rvap^{i},\rvas^{i})
$, 
and thus each of the mutual informations on the right-hand-side of~\eqref{eq:I=0+0_first} is zero.
To verify the validity of $(d)$,
we use~\eqref{eq:chainrule_I} and obtain
\begin{align}\label{eq:I=0+0}
I(\rvaq_{i+1}^{k};\rvay(i)|\rvay^{i-1},\rvaq^{i})
=
I(\rvaq_{i+1}^{k};\rvay^{i}|\rvaq^{i})
-
I(\rvaq_{i+1}^{k};\rvay^{i-1}|\rvaq^{i}),
\end{align}
where $(d)$ now follows since 
$
0\leq
I(\rvaq_{i+1}^{k};\rvay^{i-1}|\rvaq^{i})
\leq 
I(\rvaq_{i+1}^{k};\rvay^{i}|\rvaq^{i})
\leq  
I(\rvaq_{i+1}^{k};\rvay^{i},\rvar^{i},\rvap^{i}|\rvaq^{i})$, 
where the last term in this chain of inequalities was shown to be zero in the proof of $(d)$.
Equality holds in $(e)$ if and only if 
$(\rvar^{k},\rvap^{k})
\leftrightarrow
(\rvar^{i},\rvap^{i},\rvaq^{i},\rvay^{i-1})
\leftrightarrow
\rvay(i)
$, 
a Markov chain which is satisfied in our case from the fact that $(\rvaq,\rvas)\Perp(\rvar,\rvap)$ and from Lemma~\ref{lem:not_so_obvious}.

Finally, 
since $(\rvar^{k},\rvap^{k})\Perp(\rvaq^{k},\rvas^{k})$,
we have that the chain of equalities from~\eqref{eq:lamisma_I} to~\eqref{eq:penultima} holds, from which we conclude that
\begin{align}
 I(\rvar^{k},\rvap^{k};\rvay^{i}|\rvaq^{k})
=
I(\rvar^{k},\rvap^{k}\,;\, \rvau^{k})
+
I(\rvar^{k},\rvap^{k}; \rvay^{k}| \rvau^{k})
 + I(\rvaq^{k};\rvar^{k}|\rvau^{k},\rvay^{k}).
\end{align}
Inserting this result into~\eqref{eq:Irprgivenq} and invoking Theorem~\ref{thm:main} we arrive at equality $(\dagger)$ in~\eqref{eq:xtoyxtoucond}.

To prove the first equality the~\eqref{eq:xtoyxtoucond}, it suffices to notice that $I(\rvax^{k}\to\rvay^{k}|\rvaq^{k})$ corresponds to the sum on the right-hand-side of~\eqref{eq:directed_normal_cond}, from where we proceed as with the first part.
This completes the proof of the theorem.
\end{proof}

\section{Potential Applications}\label{sec:possible_applications}
Information inequalities and, in particular, the data-processing inequality, have played a fundamental role in Information Theory and its applications~\cite{zivzak73,demcov91,zamir-98,covtho06,zolalu10,weiulu11,trizam11,merhav12}.
It is perhaps the lack of a similar body of results associated with the directed information (and with non-asymptotic, causal information transmission) which has limited the extension of many important information-theoretic ideas and insights to situations involving feedback or causality constraints~\cite{naifag07,derost12}.
Two such areas, already mentioned in this paper, are the understanding of the fundamental limitations arising in networked control systems over noiseless digital channels, and causal rate distortion problems.  
In those contexts, causality is of paramount relevance an thus the directed information appears, naturally, as the appropriate measure of information flow (see, for example,~\cite{silder11,derost12,silder11b,tatiko08,elia04} and~\cite{mardah05}).  
We believe that our results might help gaining insights into the fundamental trade-offs underpinning those problems, and might also allow for the solution of open problems such as, for instance, characterizing the minimal average data-rate that guarantees a given performance level~\cite{silder10} (an improved version of the latter paper, which extensively uses the results derived here, is currently under preparation by the authors).
On a different vein, directed mutual information plays a role akin to that of (standard) mutual information when characterizing channel feedback capacity (see, e.g.,~\cite{tatmit09,li-eli11} and the references therein).
Our results may also play a role in expanding the understanding of communication problems over channels used with feedback, particularly when including in the analysis additional exogenous signals such as a random channel state, interference and, in general, any form of side information.
Thus, we hope that the inequalities and identities presented in Section~\ref{sec:results} may help in extending results such as dirty-paper coding~\cite{costa-83}, 
watermarking~\cite{cohlap02}, 
distributed source coding~\cite{slewol73,wynziv76,weiulu11,trizam11}, 
multi-terminal coding~\cite{oohama97,zamsha02},
and data encryption~\cite{johish04}, 
to scenarios involving causal feedback.

\section{Conclusions}\label{sec:conclusions}
In this paper, we have derived fundamental relations between mutual and directed informations in general discrete-time systems with feedback. 
The first of these is an inequality between the directed information between to signals inside the feedback loop 
and the mutual information involving a subset of all the exogenous incoming signals.
The latter result can be interpreted as a law of conservation of information flows for closed-loop systems. 
Crucial to establishing these bounds was the repeated use of chain rules for conditional mutual information as well as the development of new Markov chains. 
The proof techniques do not rely upon properties of entropies or distributions, and the results hold in very general cases including non-linear, time-varying and stochastic systems with arbitrarily distributed signals. Indeed, the only restriction is that all blocks within the system must be causal mappings, and that their combined delay must be at least one sample. 
A new generalized data processing inequality was also proved, which is valid for nested directed informations within the loop. 
A key insight to be gained from this inequality was that the further apart the signals are in the loop, the lower is the directed information between them. 
This closely resembles the behavior of mutual information in open loop systems, where it is well known that any independent processing of the signals can only reduce their mutual information.

\section{Appendix}
\begin{lem}\label{lem:not_so_obvious}
 In the system shown in Fig.~\ref{fig:2systems}, the exogenous signals $\rvar,\rvaq$ are mutually independent and $\Ssp_{1},\Ssp_{2}$ are deterministic (possibly time-varying) causal maps characterized by
$\rvay^{i}=\Ssp_{1}(\rvar^{i},\rvau^{i})$, 
$\rvau^{i}=\Ssp_{2}(\rvaq^{i},\rvay^{i})$, $\forall i\in\set{1,\ldots,k}$, for some $k\subset\Nl$.
\begin{figure}[htpb]
\centering
 \input{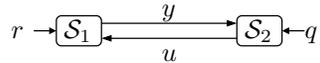}
\caption{Two arbitrary causal systems $\Ssp_{1}, \Ssp_{2}$ interconnected in a feedback loop. 
The exogenous signals $\rvar,\rvaq$ are mutually independent.}
\label{fig:2systems}
\end{figure}
For this system, the following Markov chain holds
\begin{align}
 \rvar^{k}\longleftrightarrow
(\rvau^{k},\rvay^{k})
\longleftrightarrow
\rvaq^{k},\fspace \forall k\in\mathbb{K}.
\end{align}
\end{lem}

\begin{proof}
Since 
$\rvay^{k}=\Ssp_{1}(\rvar^{k},\rvau^{k})$ 
and
$\rvau^{k}=\Ssp_{2}(\rvaq^{k},\rvay^{k})$ 
are deterministic functions, it follows that for every possible pair of sequences $y^{k},u^{k}$, the sets $\rho_{y^{k},u^{k}}\eq\set{r^{k}: y^{k}=\Ssp_{1}(r^{k},u^{k}) }$ 
and
$\phi_{y^{k},u^{k}}\eq\set{q^{k}: u^{k}=\Ssp_{2}(q^{k},y^{k}) }$ 
are also deterministic.
Thus, 
$(\rvau^{k},\rvay^{k})=(u^{k},y^{k})\iff \rvar^{k}\in\rho_{y^{k},u^{k}}$ and
$(\rvau^{k},\rvay^{k})=(u^{k},y^{k})\iff \rvaq^{k}\in\phi_{y^{k},u^{k}}$.
This means that for every pair of Borel sets $(R,Q)$ of appropriate dimensions,
\begin{align*}
\Pr\set{\rvar^{k}\in R &, \rvaq^{k}\in Q|\rvay^{k}=y^{k},\rvau^{k}=u^{k}} 
\\&
\overset{\hphantom{(a)}}{=} 
\Pr\set{\rvar^{k}\in R , \rvaq^{k}\in Q|\rvar^{k}\in\rho_{y^{k},u^{k}}\,,\, \rvaq^{k}\in\phi_{y^{k},u^{k}} }
\\&
\overset{\hphantom{(a)}}{=} 
\Pr\set{\rvar^{k}\in R |\rvar^{k}\in\rho_{y^{k},u^{k}}\,,\, \rvaq^{k}\in\phi_{y^{k},u^{k}} }
\Pr\set{\rvaq^{k}\in Q |\rvar^{k}\in (\rho_{y^{k},u^{k}}\cap R)\,,\, \rvaq^{k}\in\phi_{y^{k},u^{k}} }
\\&
\overset{(a)}{=}
\Pr\set{\rvar^{k}\in R |\rvar^{k}\in\rho_{y^{k},u^{k}}}
\Pr\set{\rvaq^{k}\in Q |\rvaq^{k}\in\phi_{y^{k},u^{k}} }
\\&
\overset{\hphantom{(a)}}{=} 
\Pr\set{\rvar^{k}\in R |\rvay^{k}=y^{k},\rvau^{k}=u^{k}}
\Pr\set{\rvaq^{k}\in Q |\rvay^{k}=y^{k},\rvau^{k}=u^{k}},
\end{align*}
where $(a)$ follows from the fact that $\rvar^{k}\Perp \rvaq^{k}$.
This completes the proof.
\end{proof}

\bibliographystyle{\BibPath/IEEEtran}

\begin{thebibliography}{10}
\providecommand{\url}[1]{#1}
\def\UrlFont{\rmfamily}
\providecommand{\newblock}{\relax}
\providecommand{\bibinfo}[2]{#2}
\providecommand\BIBentrySTDinterwordspacing{\spaceskip=0pt\relax}
\providecommand\BIBentryALTinterwordstretchfactor{4}
\providecommand\BIBentryALTinterwordspacing{\spaceskip=\fontdimen2\font plus
\BIBentryALTinterwordstretchfactor\fontdimen3\font minus
  \fontdimen4\font\relax}
\providecommand\BIBforeignlanguage[2]{{%
\expandafter\ifx\csname l@#1\endcsname\relax
\typeout{** WARNING: IEEEtran.bst: No hyphenation pattern has been}%
\typeout{** loaded for the language `#1'. Using the pattern for}%
\typeout{** the default language instead.}%
\else
\language=\csname l@#1\endcsname
\fi
#2}}

\bibitem{massey90}
J.~Massey, ``Causality, feedback and directed information,'' in \emph{Proc.~
  Intl.~Symp.~Inf.~Theory and its Appl.}, Hawaii, USA, Nov. 1990, pp. 303--305.

\bibitem{kramer98}
G.~Kramer, ``Directed information for channels with feedback.'' Ph.D.
  dissertation, Swiss federal institute of technology, 1998.

\bibitem{tatmit09}
S.~Tatikonda and S.~Mitter, ``The capacity of channels with feedback,''
  \emph{{IEEE} Trans. Inf. Theory}, vol.~55, no.~1, pp. 323--349, Jan. 2009.

\bibitem{li-eli11}
\BIBentryALTinterwordspacing
C.~{Li} and N.~{Elia}, ``The information flow and capacity of channels with
  noisy feedback,'' \emph{Submitted to {IEEE} Trans. Inf. Theory}, Aug. 2011.
  [Online]. Available: \url{http://arxiv.org/abs/1108.2815v2}
\BIBentrySTDinterwordspacing

\bibitem{derost12}
M.~S. Derpich and J.~{\O}stergaard, ``Improved upper bounds to the causal
  quadratic rate-distortion function for {G}aussian stationary sources,''
  \emph{{IEEE} Trans. Inf. Theory}, vol.~58, no.~5, pp. 3131--3152, May 2012.

\bibitem{tatiko00}
S.~C. Tatikonda, ``Control under communication constraints,'' Ph.D.
  dissertation, Department of Electrical Engineering and Computer Science,
  Massachusetts Institute of Technology, Cambridge, MA, 2000.

\bibitem{mardah05}
N.~C. Martins and M.~A. Dahleh, M., ``Fundamental limitations of performance in
  the presence of finite capacity feedback,'' in \emph{Proc.~American Control
  Conf.}, June 2005.

\bibitem{mardah08}
N.~Martins and M.~Dahleh, ``Feedback control in the presence of noisy channels:
  {``Bode-Like' '} fundamental limitations of performance,'' \emph{{IEEE}
  Trans. Autom. Control}, vol.~53, no.~7, pp. 1604--1615, Aug. 2008.

\bibitem{silder11}
E.~I. Silva, M.~S. Derpich, and J.~{\O}stergaard, ``A framework for control
  system design subject to average data-rate constraints,'' \emph{{IEEE} Trans.
  Autom. Control}, vol.~56, no.~8, pp. 1886--1899, June 2011.

\bibitem{silder10}
------, ``On the minimal average data-rate that guarantees a given closed loop
  performance level,'' in \emph{Proc.~2nd IFAC Workshop on Distributed
  Estimation and Control in Networked Systems, NECSYS}, Annecy, France, 2010,
  pp. 67--72.

\bibitem{silder11b}
------, ``An achievable data-rate region subject to a stationary performance
  constraint for {LTI} plants,'' \emph{{IEEE} Trans. Autom. Control}, vol.~56,
  no.~8, pp. 1968--1973, Aug. 2011.

\bibitem{quicol11}
C.~Quinn, T.~Coleman, N.~Kiyavash, and N.~Hatsopoulos,
  ``\BIBforeignlanguage{English}{Estimating the directed information to infer
  causal relationships in ensemble neural spike train recordings},''
  \emph{\BIBforeignlanguage{English}{Journal of Computational Neuroscience}},
  vol.~30, pp. 17--44, 2011.

\bibitem{perkim11}
H.~H. Permuter, Y.-H. Kim, and T.~Weissman, ``Interpretations of directed
  information in portfolio theory, data compression, and hypothesis testing,''
  \emph{{IEEE} Trans. Inf. Theory}, vol.~57, no.~6, pp. 3248--3259, June 2011.

\bibitem{masmas05}
J.~Massey and P.~Massey, ``Conservation of mutual and directed information,''
  in \emph{Proc.~IEEE Int.~Symp.~Information Theory}, Sept. 2005, pp. 157--158.

\bibitem{kim-yh08}
Y.-H. Kim, ``A coding theorem for a class of stationary channels with
  feedback,'' \emph{{IEEE} Trans. Inf. Theory}, vol.~54, no.~4, pp. 1488--1499,
  Apr. 2008.

\bibitem{zamkoc08}
R.~Zamir, Y.~Kochman, and U.~Erez, ``Achieving the {G}aussian rate-distortion
  function by prediction.'' \emph{{IEEE} Trans. Inf. Theory}, vol.~54, no.~7,
  pp. 3354--3364, 2008.

\bibitem{zhasun06}
H.~Zhang and Y.-X. Sun, ``Directed information and mutual information in linear
  feedback tracking systems,'' in \emph{Proc. 6-th World Congress on
  Intelligent Control and Automation}, June 2006, pp. 723--727.

\bibitem{yeung-02}
R.~W. Yeung, \emph{A first course in information theory}.\hskip 1em plus 0.5em
  minus 0.4em\relax Springer, 2002.

\bibitem{gogrsa01}
G.~C. Goodwin, S.~Graebe, and M.~E. Salgado, \emph{Control System
  Design}.\hskip 1em plus 0.5em minus 0.4em\relax New Jersey: Prentice Hall,
  2001.

\bibitem{zivzak73}
J.~Ziv and M.~Zakai, ``On functionals satisfying a data-processing theorem,''
  \emph{{IEEE} Trans. Inf. Theory}, vol.~19, no.~3, pp. 275--283, May 1973.

\bibitem{demcov91}
A.~Dembo, T.~M. Cover, and J.~A. Thomas, ``Information theoretic
  inequalities,'' \emph{{IEEE} Trans. Inf. Theory}, vol.~37, no.~6, pp.
  1501--1518, Nov. 1991.

\bibitem{zamir-98}
R.~Zamir, ``A proof of the {F}isher information inequality via a data
  processing argument,'' \emph{{IEEE} Trans. Inf. Theory}, vol.~44, no.~3, pp.
  1246--1250, May 1998.

\bibitem{covtho06}
T.~M. Cover and J.~A. Thomas, \emph{Elements of Information Theory},
  2nd~ed.\hskip 1em plus 0.5em minus 0.4em\relax Hoboken, N.J:
  Wiley-Interscience, 2006.

\bibitem{zolalu10}
J.~Zola, M.~Aluru, A.~Sarje, and S.~Aluru, ``Parallel information-theory-based
  construction of genome-wide gene regulatory networks,'' \emph{IEEE
  Transactions on Parallel and Distributed Systems}, vol.~21, no.~12, pp.
  1721--1733, Dec. 2010.

\bibitem{weiulu11}
W.~Kang and S.~Ulukus, ``A new data processing inequality and its applications
  in distributed source and channel coding,'' \emph{{IEEE} Trans. Inf. Theory},
  vol.~57, no.~1, pp. 56--69, Jan. 2011.

\bibitem{trizam11}
S.~Tridenski and R.~Zamir, ``Bounds for joint source-channel coding at high
  {SNR},'' in \emph{Proc.~IEEE Int.~Symp.~Information Theory}, Aug. 2011, pp.
  771--775.

\bibitem{merhav12}
N.~Merhav, ``Data-processing inequalities based on a certain structured class
  of information measures with application to estimation theory,'' \emph{{IEEE}
  Trans. Inf. Theory}, vol.~58, no.~8, pp. 5287--5301, Aug. 2012.

\bibitem{naifag07}
G.~N. Nair, F.~Fagnani, S.~Zampieri, and R.~J. Evans, ``Feedback control under
  data rate constraints: an overview,'' \emph{Proc. {IEEE}}, vol.~95, no.~1,
  pp. 108--137, January 2007.

\bibitem{tatiko08}
S.~Tatikonda, ``Cooperative control under communication constraints,'' in
  \emph{Proc.~Information Theory Workshop}, May 2008, pp. 243--246.

\bibitem{elia04}
N.~Elia, ``When {B}ode meets {S}hannon: Control oriented feedback communication
  schemes,'' \emph{{IEEE} Transactions on Automatic Control}, vol.~49, no.~9,
  pp. 1477--1488, 2004.

\bibitem{costa-83}
M.~Costa, ``Writing on dirty paper (corresp.),'' \emph{Information Theory, IEEE
  Transactions on}, vol.~29, no.~3, pp. 439--441, May 1983.

\bibitem{cohlap02}
A.~S. Cohen and A.~Lapidoth, ``The {G}aussian watermarking game,'' \emph{{IEEE}
  Trans. Inf. Theory}, vol.~48, no.~6, pp. 1639--1667, June 2002.

\bibitem{slewol73}
D.~Slepian and J.~Wolf, ``Noiseless coding of correlated information sources,''
  \emph{{IEEE} Trans. Inf. Theory}, vol.~19, no.~4, pp. 471-- 480, July 1973.

\bibitem{wynziv76}
A.~D. Wyner and J.~Ziv, ``The rate-distortion function for source coding with
  side information at the decoder,'' \emph{{IEEE} Trans. Inf. Theory}, vol.
  IT-22, pp. 1--10, Jan. 1976.

\bibitem{oohama97}
Y.~Oohama, ``Gaussian multiterminal source coding,'' \emph{{IEEE} Trans. Inf.
  Theory}, vol.~43, pp. 1912--1923, Nov. 1997.

\bibitem{zamsha02}
R.~Zamir, S.~Shamai, and U.~Erez, ``Nested linear/lattice codes for structured
  multiterminal binning,'' \emph{{IEEE} Trans. Inf. Theory}, no. special A.D.
  Wyner issue, pp. 1250-1276, June 2002., pp. 1250--1276, June 2002.

\bibitem{johish04}
M.~Johnson, P.~Ishwar, V.~Prabhakaran, D.~Schonberg, and K.~Ramchandran, ``On
  compressing encrypted data,'' \emph{{IEEE} Trans. Signal Process.}, vol.~52,
  pp. 2992--3006, Oct. 2004.

\end{thebibliography}

\end{document}